\documentclass{llncs}
\usepackage{amsmath}
\usepackage{amssymb}
\usepackage{epic,gastex}
\usepackage{array}
\usepackage{enumerate}

\DeclareMathOperator{\rt}{rt}
\DeclareMathOperator{\wg}{wg}

\DeclareSymbolFont{rsfscript}{OMS}{rsfs}{m}{n}
\DeclareSymbolFontAlphabet{\mathrsfs}{rsfscript}

\makeatletter
\renewcommand*{\@fnsymbol}[1]{\ifcase#1\or*\else $\dagger$\fi}
\makeatother

\begin{document}

\title{Synchronizing automata and principal eigenvectors of the underlying digraphs}
\titlerunning{Synchronizing automata and principal eigenvectors of the underlying digraphs}

\author{Vladimir V. Gusev\inst{1}\thanks{Authors were supported 
by the Russian foundation for basic research (grant 13-01-00852), Ministry of Education and Science of the Russian Federation (project no.\ 1.1999.2014/K), Presidential Program for Young Researchers (grant MK-3160.2014.1) and the Competitiveness Program of Ural Federal University.}\thanks{The author was also supported by the Belgian Programme on
Interuniversity Attraction Poles.}
\and
Elena V. Pribavkina\inst{2}\footnotemark[1]}

\authorrunning{V. V. Gusev, E. V. Pribavkina}

\tocauthor{V. V. Gusev, E. V. Pribavkina}

\institute{ICTEAM Institute, Universit{\'e} catholique de Louvain, Belgium,\\
Ural Federal University, Russia\\
\email{vl.gusev@gmail.com}
\and Institute of Mathematics and Computer Science,\\
Ural Federal University, Ekaterinburg, Russia\\
\email{elena.pribavkina@gmail.com}}

\maketitle

\begin{abstract}
A coloring of a digraph with a fixed out-degree $k$ is a distribution of $k$ labels over the edges resulting in a deterministic finite automaton.
An automaton is called synchronizing if there exists a word which sends all states of the automaton to a single state. In the present paper we study connections between spectral and synchronizing properties of digraphs. We show that if a coloring of a digraph is not synchronizing, then the stationary distribution of an associated Markov chain has a partition of coordinates into blocks of equal sum. Moreover, if there exists such a partition, then there exists a non-synchronizing automaton with such stationary distribution. 
We extend these results to bound the number of non-synchronizing colorings for digraphs with particular eigenvectors.
We also demonstrate that the length of the shortest synchronizing word of any coloring is at most $w^2 - 3w + 3$, where $w$ is the sum of the coordinates of the integer dominant eigenvector of the digraph.
\end{abstract}

\section{Introduction}
Let $\mathrsfs{A}=(Q,\Sigma,\delta)$ be a deterministic
finite automaton over a finite alphabet $\Sigma$ with the set of states $Q$
and the transition function $\delta: Q\times\Sigma\mapsto Q$. The function $\delta$ naturally extends to the subsets of $Q$, and the the free monoid $\Sigma^*$. When $\delta$ is clear from the context we will write $qw$ for $\delta(q,w)$ where $q$ is a state and $w$ is a word over the alphabet $\Sigma$. The automaton $\mathrsfs{A}$ is called \emph{synchronizing}
if there exist a word $w$ and a state $p$ such that for every state $q \in Q$ we have
$qw=p$. Any such word $w$ is called \emph{synchronizing} (or \emph{reset})
word for $\mathrsfs{A}$. The length of the shortest synchronizing word $\rt(\mathrsfs{A})$ is called \emph{the reset threshold} of $\mathrsfs{A}$.
A good survey of the theory of synchronizing automata may be found in~\cite{Volkov2008Survey}. One of the oldest problems in this area is the \v{C}ern\'{y} conjecture that states that the reset threshold of an $n$-state automaton is at most $(n-1)^2$~\cite{Cerny1964}. Our results will have some implications for the conjecture.

\emph{The underlying digraph} of an automaton $\mathrsfs{A}$ is a digraph  with $Q$ as a set of vertices, and whose set of edges consists of pairs $(u,v)$ such that $\delta(u,x)=v$ for some letter $x\in\Sigma$. Note that, such a digraph has a fixed out-degree (equal to the cardinality of the alphabet $\Sigma$), loops and multiple edges are allowed. Vice versa, given a digraph $G$ with a fixed out-degree $k$ and a finite alphabet $\Sigma$ with $k$ letters, we can obtain a deterministic finite automaton by distributing the letters of $\Sigma$ over the edges of $G$. Any automaton obtained in this way is called \emph{a coloring} of $G$. A digraph is \emph{primitive} if there exists a number $t$ such that for any two vertices $u$ and $v$ there exists a path from $u$ to $v$ of length exactly $t$. An automaton is \emph{strongly connected} if its underlying digraph is strongly connected. It is well known that the underlying digraph of a strongly connected synchronizing automaton is primitive. Furthermore, the converse is also true, every primitive digraph has a synchronizing coloring~\cite{Tr2009RoadColoring}. Previously, the latter statement was a conjecture known under the name ``the road coloring problem''. It was open for 30 years and received a lot of attention before it was finally resolved by Trahtman~\cite{Tr2009RoadColoring} after the crucial insight by Culik, Karhum\"{a}ki, and Kari~\cite{CuKaKa2002}. Recently, there were suggested two directions to generalize the road coloring theorem~\cite{GuSz2015}. The first of them is related to the synchronizing properties of a random primitive digraph of a fixed out-degree. A digraph is called \emph{totally synchronizing} if all of its colorings are synchronizing. Relying on computational experiments the authors of~\cite{GuSz2015} conjectured that a random digraph is totally synchronizing. The other direction involves the fraction of synchronizing colorings among all possible colorings. The same authors conjectured that this fraction is at least $\frac{k-1}{k}$ for a primitive digraph with a fixed out-degree $k$. If both of these conjectures are true, then the road coloring theorem is a relatively weak statement that gives us just the first step to satisfactory understanding of the synchronizing properties of digraphs.

One of the major difficulties in dealing with the colorings of digraphs is checking whether they are synchronizing or not. In the present paper we rely on the integer dominant eigenvector of a digraph to analyze its colorings. We show that a $k$-out-regular digraph is totally synchronizing if coordinates of the integer dominant eigenvector cannot be partitioned into blocks of equal sum. Moreover, the eigenvector $\vec{w}$ is not partitionable if and only if all the digraphs which have $\vec{w}$ as the eigenvector are totally synchronizing. Furthermore, if the partition of the coordinates is unique, then the fraction of synchronizing colorings is at least $\frac{k-1}{k}$, i.e. the generalized road coloring problem holds true in this case. Furthermore, we demonstrate that if the sum of the coordinates is equal to $w$, then the reset threshold is at most $w^2 - 3w + 3$ which has some implications for the \v{C}ern\'{y} conjecture. We believe that our techniques can be used to advance in the generalizations of the road coloring theorem. Furthermore, they shed light on the intriguing connections between spectral and synchronizing properties of digraphs. 

\section{The eigenvectors of synchronizing automata}
Let $\mathrsfs{A}$ be a strongly connected automaton with the set of states $\{1,2, \ldots, n\}$. Let $A_1, A_2, \ldots, A_k$ be the adjacency matrices of the letters of $\mathrsfs{A}$, i.e. $A_\ell[i,j] =1$ if and only if $i$ is mapped to $j$ under the action of $\ell$th letter. Consider a matrix $A=\sum_{i=0}^{k} p_i A_i$, where $p_i >0$ are rational for all $i$ and $\sum_{i=0}^{k} p_i = 1$. Since the matrix $A$ is row-stochastic the largest eigenvalue of $A$ is equal to 1. By the Perron-Frobenius theorem there exists a left positive eigenvector $\vec{u}$ such that $\vec{u}A=\vec{u}$. Since the entries of $A$ are rational so are the entries of $\vec{u}$. Let $\vec{w} = \ell \vec{u}$, where $\ell$ is the least common multiple of the denominators of entries of $\vec{u}$. We will call the vector $\vec{w}$ \emph{the eigenvector} of $\mathrsfs{A}$ in accordance with the distribution $p_1, \ldots, p_k$. If the distribution is uniform, i.e. $p_1= p_2 =\ldots =p_k$, then we will usually omit its description. Since all colorings of a digraph $G$ have the same eigenvector $\vec{w}$ in accordance with the uniform distribution we will call $\vec{w}$ obtained in this way \emph{the eigenvector} of $G$.

The following theorem by Friedman~\cite{Friedman1990OnTheRoadColoringProblem} is the starting point of our paper. Given an automaton $\mathrsfs{A}=(Q,\Sigma,\delta)$, \emph{a kernel} of a word $x$ with respect to $\mathrsfs{A}$ is an equivalence relation $\rho$ on the set of states $Q$ such that $i \rho j$ if and only if $ix = jx$. A subset $S$ is \emph{synchronizing} if there exists a word $x$ such that $Sx=1$. By $Sx^{-1}$ we denote the full preimage of the set $S$ under the action of a letter $x$, i.e. $Sx^{-1}=\{q\in Q\mid qx\in S\}$. Let $\vec{w}$ be the eigenvector of the automaton $\mathrsfs{A}$. We define \emph{the weight} $\wg (i)$ of a state $i$ as $\vec{w}[i]$. The weight of a set $S$ is defined as $\wg(S) = \sum_{i \in S} \wg(i)$.
\begin{theorem}
\label{th:friedman}
Let $\vec{w}$ be the eigenvector of an automaton $\mathrsfs{A}$.
There is a partition of the states of $\mathrsfs{A}$ into synchronizing subsets of maximal weight. Furthermore, this partition is equal to the kernel of some word $x$.
\end{theorem}

We will prove a generalization of this theorem.
\begin{theorem}
\label{th:partition}
Let $\vec{w}$ be the eigenvector of an automaton $\mathrsfs{A}$ in accordance with a distribution $p_1, p_2, \ldots, p_k$. There exists a partition of the states of $\mathrsfs{A}$ into synchronizing subsets of maximal weight. Furthermore, this partition is equal to the kernel of some word $x$.
\end{theorem}
\begin{proof}
Let $\Sigma=\{a_1,a_2,\ldots,a_k\}$, and let $S$ be an arbitrary subset of $Q$. Note the following equality:
$$\sum_{i=0}^{k} p_i \wg( S a_i^{-1}) = \wg(S)$$ (the incoming edges to $S$ bring in total the weight equal to $\wg(S)$, and each preimage brings $p_i \wg( S a_i^{-1})$). If $S$ is a synchronizing subset of maximal weight, then the
weights of preimages are bounded by $\wg(S)$, since every preimage is also a synchronizing subset. Moreover, every preimage has the weight equal to $\wg(S)$, otherwise the left-hand side would be strictly less than the right-hand side. Therefore, if $S$ is a synchronizing subset of maximal weight, then every preimage of $S$ is a synchronizing subset of maximal weight.

We will iteratively construct a partition of the set of states of $\mathrsfs{A}$ into synchronizing subsets of maximal weight. Let $S_0$ be a synchronizing subset of maximal weight. Let $u$ be a word synchronizing $S_0$ to some state $q$: $S_0 u = q$. If $S_0=Q$, then the automaton is synchronizing, and the proof is complete. Otherwise, let $p$ be a state that doesn't belong to $S_0$. Since the automaton $\mathrsfs{A}$ is strongly connected, there exists a word $v$ such that $q v = p$. Consider now the sets $S_1 = S_0(uv)^{-1}$ and $S_0$. Note, that $S_1$ is also a maximal synchronizing subset by the preceding paragraph. Both sets are synchronized by $uvu$. But their images are different, since $q$ is not equal to $pu$ due to maximality of $S_0$. Continuing in the same manner we will eventually construct the desired partition of $Q$. \qed
\end{proof}

The fact that we assign probabilities to the letters of an automaton $\mathrsfs{A}$ in the statement of theorem~\ref{th:partition}, but not to the edges of its underlying digraph, is necessary. Let $\mathrsfs{F}$ be an automaton depicted in Fig.~\ref{fig:flip-flop}. The notation $\ell / p$ means that the edge is labelled by $\ell$ and has the probability $p$. Note, that the eigenvector of $\mathrsfs{F}$ is equal to $(1-p_2,1-p_1)$.
Since every letter acts as a permutation, the automaton $\mathrsfs{F}$ is not synchronizing. Therefore, the partition of the states into synchronizing subsets should be of the form $\{\{0\}, \{1\}\}$, but
for $p_2=\frac{1}{3}$ and $p_1=\frac{1}{2}$ these subsets have different weight.
\begin{figure}[ht]
 \begin{center}
  \unitlength=4pt
    \begin{picture}(20,15)(0,0)
    \gasset{Nw=6,Nh=6,Nmr=6}
\node(A0)(0,8){$0$}
\node(A1)(20,8){$1$}
    \drawedge[curvedepth=4](A0,A1){$b/1-p_1$}
    \drawedge[curvedepth=4](A1,A0){$b/1-p_2$}
    \drawloop[loopangle=180](A0){$a/p_1$}
    \drawloop[loopangle=0](A1){$a/p_2$}
    \end{picture}
 \end{center}
  \caption{An automaton $\mathrsfs{F}$}
  \label{fig:flip-flop}
\end{figure}

\begin{corollary}
Let $\vec{w}$ be the eigenvector of an automaton $\mathrsfs{A}$ in accordance with a distribution $p_1, p_2, \ldots, p_k$. 
If there is no partition of $\vec{w}$ into blocks of equal weight, then the automaton $\mathrsfs{A}$ is synchronizing.
\end{corollary}
Unfortunately, the converse of this corollary does not hold. Let $\mathrsfs{B}$ be an automaton depicted in Fig.~\ref{fig:B}. It is synchronized by the word $bbaab$ to the state 1. If $p$ and $1-p$ are the probabilities of the letters $a$ and $b$ respectively, then the eigenvector of $\mathrsfs{B}$ is equal to $(1,1,p,p)$. Thus, the subsets $\{0,2\}$ and $\{1,3\}$ form a partition of the states of $\mathrsfs{B}$ for any $p$, in other words, there is no witness of the fact that $\mathrsfs{B}$ is synchronizing.
\begin{figure}[ht]
 \begin{center}
  \unitlength=4pt
    \begin{picture}(20,24)(0,0)
    \gasset{Nw=4,Nh=4,Nmr=2}
\node(A0)(0,20){$0$}
\node(A1)(20,20){$1$}
\node(A2)(20,0){$2$}
\node(A3)(0,0){$3$}
    \drawedge[curvedepth=2](A0,A1){$b$}
    \drawloop[loopangle=-180,loopdiam=5](A0){$a$}
    \drawedge[curvedepth=2](A1,A0){$b$}
    \drawedge(A1,A3){$a$}
    \drawedge[curvedepth=2](A2,A1){$a$}
    \drawedge[curvedepth=-2](A2,A1){$b$}
    \drawedge[curvedepth=2](A3,A2){$a$}
    \drawedge[curvedepth=-2](A3,A2){$b$}
    \end{picture}
 \end{center}
  \caption{An automaton $\mathrsfs{B}$}
  \label{fig:B}
\end{figure}

\section{The eigenvectors and the reset thresholds}

Surprisingly, the knowledge of the eigenvector also allows us to infer an upper bound on the reset threshold. 
In the main result of this section we will transform a given synchronizing automaton into an Eulerian one, and make use of the following theorem by Kari~\cite[Theorem 2]{Kari2003Eulerian}:

\begin{theorem}
\label{th:euler}
The reset threshold of an Eulerian $n$-state automaton is at most $n^2 - 3n + 3$.
\end{theorem}

\begin{theorem}
\label{th:eigenrt}
Let $\mathrsfs{A}$ be a synchronizing automaton and $\vec{w}$ be the eigenvector in accordance with a distribution $p_1,p_2, \ldots, p_k$. The reset threshold of $\mathrsfs{A}$ is at most $w^2 - 3w + 3$, where $w$ is the sum of the entries of $\vec{w}$.
\end{theorem}

\begin{proof}
Let $\Sigma = \{a_1, a_2, \ldots, a_k \}$ and $p_i = \frac{m_i}{\ell}$ for $1\leq i \leq k$, where $m_i, \ell$ are positive integers. First, we will duplicate some edges in order to reduce our problem to automata with the eigenvectors obtained in accordance with the uniform distribution.  We can construct a new automaton $\mathrsfs{A}'$ on the set of states of $\mathrsfs{A}$ with the alphabet $\Sigma' = \{ a_1^1, a_1^2, \ldots, a_1^{m_1}, a_2^1, a_2^2, \ldots, a_2^{m_2}, \ldots , a_k^1, a_k^2, \ldots, a_k^{m_k} \}$. For arbitrary $i$ and $j$, the action of the letter $a_i^j$ in $\mathrsfs{A}'$ coincides with the action of the letter $a_i$ in $\mathrsfs{A}$. It is easy to see that $\mathrsfs{A}'$ is synchronizing and $\rt(\mathrsfs{A}') = \rt(\mathrsfs{A})$. Furthermore, the eigenvector of $\mathrsfs{A}'$ in accordance with a uniform distribution coincides with $\vec{w}$.

Now we are going to construct an Eulerian automaton $\mathrsfs{A}''$ on a larger set of states with the reset threshold at least $\rt(\mathrsfs{A}')$. We will also slightly enlarge the alphabet of $\mathrsfs{A}'$. Let $Q = \{0,1, \ldots, n-1\}$ be the set of states of $\mathrsfs{A}'$ and $\vec{w}[i]$ be the $i$th entry of $\vec{w}$. The set of states of $\mathrsfs{A}''$ is equal to $\{(i,j)\, |\, i \in Q, 1 \leq j \leq \vec{w}[i] \}$. Let $S_i$ be the set $\{ (i,j) \, |\, 1 \leq j \leq \vec{w}[i]\}$. We define the action of every letter $x \in \Sigma'$ in such a way that $S_i x \subseteq S_{ix}$, where $ix$ is an image of $i$ in $\mathrsfs{A}'$. Furthermore, we will require that $\mathrsfs{A}''$ is an Eulerian automaton, i.e. the in-degree of every state is equal to the out-degree. Let us show that it is always possible to achieve these goals. Let $k' = |\Sigma'|$ and $c_{ij}$ be the number of letters in $\Sigma'$ that bring $i$ to $j$ in $\mathrsfs{A}'$. By the definition of $\vec{w}$ we have $\sum_{i \in Q} c_{ij} \vec{w}[i] = k' \vec{w}[j]$ for every $j$. In other words, the number of incoming edges to $S_j$ is equal to $k' |S_j|$. Thus, we can redistribute the incoming edges to $S_j$ in such a way that every state has in-degree equal to $k'$. Since the out-degree of $\mathrsfs{A}'$ is equal to $k'$ we get that $\mathrsfs{A}'$ is Eulerian.
We will also add an extra set of letters $\Lambda$ to the automaton $\mathrsfs{A}''$ in the following way. For every $i \in Q$ and $j \in S_i$ we add a letter that brings all the states from $S_i$ to the state $j$, while all the other states are fixed. Note, that the automaton $\mathrsfs{A}''$ over the enlarged alphabet $\Sigma' \cup \Lambda$ is still Eulerian.

Now we are ready to show that $\rt (\mathrsfs{A}') = \rt (\mathrsfs{A})$ is at most $\rt (\mathrsfs{A}'')$. Let $u$ over the alphabet $\Sigma' \cup \Lambda$ be a synchronizing word of $\mathrsfs{A}''$ that brings it to a state in $S_i$ for some $i$. Let $v$ be a word over the alphabet $\Sigma'$ obtained from $u$ by removing all the letters from $\Lambda$. Since the action of every letter $x$ from $\Sigma'$ of the automaton $\mathrsfs{A}''$ satisfies $S_i x \subseteq S_{ix}$, where $ix$ is the image of $i$ in the automaton $\mathrsfs{A}'$, we conclude that $u$ is a synchronizing word for the automaton $\mathrsfs{A}'$ and $\rt (\mathrsfs{A}) \leq \rt (\mathrsfs{A}'')$.

Note, that the automaton $\mathrsfs{A}''$ has $w$ states. Therefore, by theorem~\ref{th:euler} we conclude that the reset threshold of $\mathrsfs{A}''$ is at most $w^2 -3w +3$. \qed
\end{proof}
\begin{remark}
The eigenvector of a digraph $G$ is the same for every coloring, thus the bound presented in theorem~\ref{th:eigenrt} holds true for every coloring.
\end{remark}

The theorem~\ref{th:eigenrt} gives a simple combinatorial proof of the quadratic upper bound for the reset threshold in the class of pseudo-eulerian automata~\cite{Berlinkov2013QuasiEulerianOneCluster,Steinberg2011AveragingTrick}.
Generally speaking, we can formulate the following corollary:
\begin{corollary}
Let $\mathcal{C}$ be a class of synchronizing automata. If the sum of the entries of the eigenvector of each $n$-state automaton $\mathrsfs{C} \in \mathcal{C}$ in accordance with some distribution is bounded by $f(n)$, then the reset threshold is bounded by $f^2(n) -3 f(n) + 3$.
\end{corollary}

In general, the upper bound in the theorem~\ref{th:eigenrt} can be exponential in terms of the number of states. Consider the following $k$-out-regular digraph $U_{n,k}$. The set of vertices is equal to $\{0,1,2,\ldots , n-1\}$. For each $0\leq i \leq n-1$ there is an edge $(i,0)$ of multiplicity $k-1$ and an edge $(i,i+1 \mod n)$. It is easy to verify that the integer eigenvector of $U_{n,k}$ is  $(k^{n-1}, k^{n-2}, \ldots, k, 1)$. Thus, the sum of the entries is $\frac{k^n - 1}{k-1}$. 

\begin{proposition}
Given a $k$-out-regular digraph $G$ with $n$ vertices. The entries of the eigenvector $\vec{w}$ are upper-bounded by $(2k^2)^{\frac{n-1}{2}}$.
\end{proposition}
\begin{proof}
The eigenvector $\vec{w}$ satisfies the equation: $\vec{w}(A - k I) = 0$, where $A$ is the adjacency matrix of $G$ and $I$ is the identity matrix. Thus, the rank of $A -kI$ is at most $n-1$.
The main result of~\cite{Borosh1977} states that for every integer matrix $M$ of rank $r$ if a system $Mx=0$ admits a nontrivial non-negative integer solution, then there exists such a  solution with entries bounded by the maximum of the absolute values of the $r \times r$ minors of $M$.

Thus, we conclude that there exists a non-negative integer vector $\vec{w'}$ such that $\vec{w}(A - k I) = 0$ and entries of $\vec{w}'$ are bounded by the maximum of the absolute values of the $(n-1)\times (n-1)$ minors of $A-kI$. Note, that the norm of each row of every minor is at most $\sqrt{2k^2}$. Thus, by Hadamard's inequality for the determinant we obtain an upper bound $(2k^2)^{\frac{n-1}{2}}$ on the minors. Since the non-negative vector $\vec{w'}$ is an eigenvector of $A$ associated with the largest eigenvalue, we immediately conclude that $\vec{w'}$ is positive by the Perron-Frobenius theorem.
\end{proof}

\section{The eigenvectors of totally synchronizing digraphs}

Let $\vec{w}$ be a positive integer vector. In this section we will call an equivalence relation $\beta$ on the coordinates of $\vec{w}$ \emph{a partition} if it satisfies the following property: there exists a constant $b$ such that for every class $B$ of $\beta$ we have $\sum_{i \in B} \vec{w}[i]=b$. We will refer to the classes of a partition $\beta$ as \emph{blocks}. If $\vec{w}$ is the eigenvector of an automaton $\mathrsfs{A}$, then every coordinate correspond to the state of $\mathrsfs{A}$. Thus, we can naturally obtain an equivalence relation $\beta'$ on the states of $\mathrsfs{A}$ from the partition $\beta$. Abusing notation, we will refer to $\beta'$ as $\beta$. A vector $\vec{w}$ is called \emph{partitionable} if it possesses a partition.
Let $\mathcal{G}(\vec{w})$ be a class of primitive digraphs with the eigenvector $\vec{w}$ such that every digraph in the class has a fixed out-degree (which can be different for two different digraphs from the class).
\begin{theorem}
\label{pr:blocks}
A vector $\vec{w}$ is not partitionable if and only if all digraphs from $\mathcal{G}(\vec{w})$ are totally synchronizing.
\end{theorem}
\begin{proof}
Let $G$ be a digraph from $\mathcal{G}(\vec{w})$. If $G$ has a non-synchronizing coloring, then by theorem~\ref{th:partition} it admits a partition of the states into synchronizing subsets of equal weight. Thus, the vector $\vec{w}$ is also partitionable.

Assume now that $\vec{w}$ is partitionable, i.e. there are sets $B_1, B_2, \ldots, B_\ell$ such that for every $i$ we have $\sum_{j \in B_i} \vec{w}[j] = b$. We will construct a primitive digraph $G$ on the set of vertices $V=\{0,1, \ldots, n-1\}$, where $n$ is the number of entries in $\vec{w}$, that is not totally synchronizing. The set of edges is defined as follows: for every $i$ and $j$ there is an edge $(i,j)$ of multiplicity $\vec{w}[j]$. Note, that the out-degree of every vertex is equal to the total weight of $\vec{w}$, i.e. $b\ell$. Furthermore, the digraph $G$ is primitive. We have $\sum_{i \in V} c_{ij} \vec{w}[i]=\sum_{i \in V} \vec{w}[j] \vec{w}[i]=\vec{w}[j] \sum_{i \in V} \vec{w}[i] = b \ell \vec{w}[j]$, where $c_{ij}$ is the multiplicity of the edge from $i$ to $j$. Therefore, $\vec{w}$ is the eigenvector of $G$.

Now we are going to construct a non-synchronizing coloring of $G$. Let $\mathcal{A}$ be a set of colorings of $G$
that satisfy the following property: for every letter $x$ and for every $i,j \in V$ such that $i,j \in B_s$ for some $s$ if $ix \in B_t$, then $j x \in B_t$.
For every automaton $\mathrsfs{A} \in \mathcal{A}$ the partition of the states $\beta = \{B_1, B_2, \ldots, B_\ell \}$ is a congruence. Let's fix some automaton $\mathrsfs{A} \in \mathcal{A}$. The factor automaton $\mathrsfs{A'}$ of $\mathrsfs{A}$ with respect to $\beta$ is Eulerian. Lemma~1 from~\cite{Kari2003Eulerian} states that every Eulerian automaton has a non-synchronizing coloring\footnote{It is also a simple consequence of the Birkhoff-von Neumann theorem}. Thus, we can recolor an automaton $\mathrsfs{A'}$ into a non-synchronizing automaton $\mathrsfs{A''}$. It is not hard to see, that $\mathrsfs{A''}$ is a factor automaton of some automaton $\mathrsfs{B} \in \mathcal{A}$. Furthermore, the automaton $\mathrsfs{B}$ is not synchronizing, since its factor automaton is not synchronizing. \qed
\end{proof}

Theorem~\ref{pr:blocks} allows us to obtain very simple proofs for otherwise non-obvious statements. Recall that the \v{C}ern\'{y} automaton $\mathrsfs{C}_n$ \cite{Cerny1964} can be defined as $\langle \{0,\ldots,n-1\},\{a,b\},\delta \rangle$, where $\delta(i,a)=i+1$ for $i<n-1$, $\delta(n-1,a)=0$, $\delta(n-1,b)=0$, and $\delta(i,b)=i$ for $i<n-1$.
\begin{proposition}{\cite[Proposition 2]{GuSz2015}}
The underlying digraph of the \v{C}ern\'{y} automaton $\mathrsfs{C}_n$ is totally synchronizing.
\end{proposition}
\begin{proof}
It is easy to verify that the eigenvector $\vec{w}$ of the underlying digraph of the $n$-state \v{C}ern\'{y} automaton is equal to $(2,2,\ldots, 2, 1)$. Since in every partition exactly one block will have an odd sum, we conclude that  $\vec{w}$ is not partitionable. Thus, the digraph is totally synchronizing.
\end{proof}
A similar proof can be presented for many other examples in~\cite{AGV2013}. 
We want to remark the following easy to use corollary:
\begin{corollary}
If $G$ is not an Eulerian digraph and its eigenvector is not partitionable, then $G$ is totally synchronizing.
\end{corollary}

There are classes of digraphs $\mathcal{G}(\vec{w})$ that contain both totally synchronizing and not totally synchronizing digraphs. Let $\vec{w}$ be $(1,1,2,2)$. The underlying digraph of the automaton $\mathrsfs{D}$, see fig.~\ref{fig:D}, belongs to $\mathcal{G}(\vec{w})$. It is not totally synchronizing, since the pair $\{2,3\}$ is not synchronizable in the coloring $\mathrsfs{D}$. At the same time, it is easy to see that the underlying digraph of the automaton $\mathrsfs{B}$, see fig.~\ref{fig:B}, belongs to $\mathcal{G}(\vec{w})$ and it is totally synchronizing.

There are also classes of digraphs $\mathcal{G}(\vec{w})$ which do not contain totally synchronizing digraphs at all. Namely, if $\vec{w} = (1,1,\ldots,1)$ then every digraph in $\mathcal{G}(\vec{w})$ is Eulerian, thus it possesses a non-synchronizing coloring~\cite[Lemma 1]{Kari2003Eulerian}.
\begin{figure}[ht]
 \begin{center}
  \unitlength=4pt
    \begin{picture}(20,20)(0,0)
    \gasset{Nw=4,Nh=4,Nmr=2}
\node(A0)(0,20){$3$}
\node(A1)(20,20){$0$}
\node(A2)(20,0){$1$}
\node(A3)(0,0){$2$}
    \drawedge[curvedepth=2](A0,A3){$b$}
    \drawloop[loopangle=-180](A0){$a$}
    \drawedge[curvedepth=0](A1,A3){$b$}
    \drawloop[loopangle=0](A1){$a$}
    
    \drawedge[curvedepth=0](A2,A1){$b$}
    \drawedge[curvedepth=2](A2,A3){$a$}
    \drawedge[curvedepth=2](A3,A2){$a$}
    \drawedge[curvedepth=2](A3,A0){$b$}
    \end{picture}
 \end{center}
  \caption{Automaton $\mathrsfs{D}$}
  \label{fig:D}
\end{figure}

Recall the following non-trivial conjecture made in~\cite[Conjecture 3]{GuSz2015}:
\begin{conjecture}
For every $k \geq 2$, the fraction of totally synchronizing digraphs
among all $k$-out-regular digraphs with $n$ vertices tends to 1 as $n$ goes to infinity.
\end{conjecture}
We believe that significant progress on this conjecture can be made through the study of the eigenvectors of digraphs. Despite the fact that the statement of theorem~\ref{pr:blocks} gives only a necessary condition for a digraph to be totally synchronizing we expect it to hold in majority of cases. Namely, we propose the following conjecture:
\begin{conjecture}
The eigenvector of a random $k$-out-regular digraph with $n$ vertices has no partition into blocks of equal sum with probability 1 as $n$ goes to infinity.
\end{conjecture}

Recall that the \emph{synchronizing ratio} of a $k$-out-regular digraph $G$ is ratio of the number of synchronizing colorings to the number $(k!)^n$ of all possible colorings. In the remainder of this section we are going to shed some light on another conjecture made in~\cite[Conjecture 1]{GuSz2015}:
\begin{conjecture}
\label{conj:ratio}
The minimum value of the synchronizing ratio among all $k$-out-regular digraphs with $n$ vertices is equal to $\frac{k-1}{k}$, except for the case $k=2$ and $n=6$ when it is equal to $\frac{30}{64}$.
\end{conjecture}

\begin{lemma}
\label{lem:unique}
Let $\mathrsfs{A}$ be a non-synchronizing automaton with the eigenvector $\vec{w}$ in accordance with a distribution $p_1,p_2, \ldots, p_k$. Let $b$ be the weight of a maximal synchronizing subset. The block partition of the eigenvector $\vec{w}$ into subsets of weight $b$ is unique if and only if the partition into maximal synchronizing subsets of weight $b$ is a congruence.
\end{lemma}
\begin{proof}
Assume first that the partition into blocks of $\vec{w}$ is unique. We will denote the block of an element $p$ by $[p]$. If the partition is not a congruence, then there exists a letter $\ell$ such that $[p]=[q]$ and $[p \ell] \neq [q \ell]$ for some states $p$ and $q$. Note, that the preimage of a maximal synchronizing subset is also a maximal synchronizing subset (see the proof of theorem~\ref{th:partition}). Moreover, the preimage of a partition into maximal synchronizing subsets is also a partition into maximal synchronizing subsets. Thus, $[p\ell] \ell^{-1}$ is a maximal synchronizing subset and $[p\ell]\ell^{-1} \cap [p] \neq \varnothing$. We also have $[p \ell ] \ell^{-1} \neq [p]$, otherwise $[p \ell]=[p]\ell$ which implies $[p \ell] = [q \ell]$. Therefore, the preimage of the partition by the letter $\ell$ is a different partition into maximal synchronizing subsets. A contradiction.

Assume now that the partition $\tau$ is a congruence. Assume to the contrary that there is another partition $\sigma$ into synchronizing subsets of maximal weight. Note, that there are states $p$ and $q$ such that $p \sim_{\sigma} q$ and $p \nsim_{\tau} q$, otherwise $\sigma$ is a refinement of $\tau$, and $\sigma$ is not a partition into synchronizing subsets of maximal weight. Since $p \sim_{\sigma} q$ there exists a word $u$ such that $p u = q u$. Let $[p]$ and $[q]$ be the blocks of the partition $\tau$ of $p$ and $q$ respectively. Since $\tau$ is a congruence both $[p] u$ and $[q] u$ are subsets of the same block $[r]$ for some state $r$. The subset $[r]$ is synchronizing. Therefore, the subset $[p] \cup [q]$ is also synchronizing, which contradicts maximality of $[p]$ and $[q]$.
\end{proof}
\begin{corollary}
\label{co:congruence}
A digraph $G$ with the eigenvector $\vec{w}$ is totally synchronizing if the following statements hold:
\begin{enumerate}
\item if there exists a partition of $\vec{w}$ into blocks of weight $b$, then it is unique;
\item every partition of $\vec{w}$ is not a congruence for every coloring.
\end{enumerate}
\end{corollary}
Using this statements we can obtain a result that supports conjecture~\ref{conj:ratio}:
\begin{theorem}
\label{th:half}
If the block partition of the eigenvector $\vec{w}$ is unique, then the synchronizing ratio of every $k$-out-regular digraph in $\mathcal{G}(\vec{w})$ is at least $\frac{k-1}{k}$.
\end{theorem}
\begin{proof}
Let $\beta$  be the unique block partition of $\vec{w}$ and $G$ be a digraph in  $\mathcal{G}(\vec{w})$.
Assume first that the partition consists only of singletons. It implies that $\vec{w} = (1,1 \ldots, 1)$ and $G$ is an Eulerian digraph with a prime number of vertices. Note, that there exist a vertex $q$ and edges $(q,r),(q,s)$ for $r\neq s$, otherwise $G$ is not primitive. Let all the $k$ outgoing edges of $q$ be $(q, p_1)$ of multiplicity $k_1$, $(q, p_2)$ of multiplicity $k_2$, $\ldots$, $(q, p_\ell)$ of multiplicity $k_\ell$. Let $\mathcal{A}$ be a set of colorings of $G$ that have the same labelling of all edges, except for the outgoing edges of $q$. In order to show that the synchronizing ratio of $G$ is at least $\frac{k-1}{k}$ we will demonstrate that the fraction of non-synchronizing automata in $\mathcal{A}$ is at most $\frac{1}{k}$. If all automata in $\mathcal{A}$ are synchronizing, then the statement holds true. Otherwise, let $\mathrsfs{A} \in \mathcal{A}$ be a non-synchronizing automaton. Since the number of states is prime we conclude by theorem~\ref{th:partition} that every letter of $\mathrsfs{A}$ acts as a permutation on the set of states. Note, that if edges $(q,p_1)$ and $(q,p_2)$ are labelled by $x$ and $y$ respectively, then the automaton $\mathrsfs{A'}$ obtained by flipping the labels on these edges, i.e. assigning letter $y$ to $(q,p_1)$ and letter $x$ to $(q,p_2)$, is synchronizing. 
Indeed, either $p_1$ or $p_2$ is not equal to $q$, without loss of generality we will assume that $p_1 \neq q$. Since every letter in $\mathrsfs{A}$ acts as a permutation, there exists a state $r$ such that $r y = p_1$. Thus, $r y = q y$ for the automaton $\mathrsfs{A'}$ and it is synchronizing by theorem~\ref{th:partition}. More generally, there are at most $k_1!k_2!\ldots k_\ell !$ permutations of labels on the outgoing edges of $q$ that keep the resulting automaton non-synchronizing. Since the value of the fraction $\frac{k_1!k_2!\ldots k_\ell !}{k!}$ is bounded by $\frac{1}{k}$ we obtain the desired statement.

Assume now that for every coloring $\mathrsfs{A}$ of $G$ the partition $\beta$ is a congruence. It implies that for every non-singleton block $B$ there is a block $B'$ such that all the outgoing edges of $B$ arrive to $B'$. Thus, $\beta$ has a singleton block, otherwise $G$ is a union of cycles which can't be primitive. We easily obtain the following statement: a coloring $\mathrsfs{A}$ of $G$ is synchronizing if and only if the factor automaton $\mathrsfs{A}/\beta$ is synchronizing. Note, that the factor automaton $\mathrsfs{A}/\beta$ is Eulerian with a prime number of states. Thus, by the preceding case we conclude that the synchronizing ratio is at most $\frac{k-1}{k}$.

Assume now that there exists a coloring $\mathrsfs{A}$ such that $\beta$ is not a congruence. Thus, there are states $q,p$ from the same block $B$ and a letter $\ell$ such that $q\ell$ and $p\ell$ belong to different blocks of $\beta$. If for all colorings of $G$ the partition $\beta$ is not a congruence, then $G$ is totally synchronizing by corollary~\ref{co:congruence} and the statement of the theorem holds. Thus, we assume that there is a coloring which is a congruence for $\beta$. It implies that for every block $B'$ there is the same number of edges going from $q$ to $B'$ and from $p$ to $B'$. Let $k_1$ be the number of edges going from $q$ to $B_1$, $k_2$ be the number of edges going from $q$ to $B_2$, \ldots, $k_\ell$ be the number of edges going from $q$ to $B_\ell$. Let $\mathcal{A}$ be a set of colorings of $G$ that have the same labelling of all edges, except for the outgoing edges of $q$ and $p$. We will show that the fraction of non-synchronizing automata in $\mathcal{A}$ is at most $\frac{1}{k}$. It is not hard to see, that there are at most ${k \choose k_1}(k_1!)^2{k-k_1 \choose k_2}(k_2!)^2\ldots {k_\ell \choose k_\ell}(k_\ell !)^2$ automata in $\mathcal{A}$ for which $\beta$ is a congruence. If $\beta$ is not a congruence for a given coloring, then by lemma~~\ref{lem:unique} we conclude that this coloring is synchronizing. Therefore, after the simplification we conclude that the fraction of non-synchronizing automata is at most $\frac{k_1!k_2!\ldots k_\ell !}{k!}$. Thus, it is bounded by $\frac{1}{k}$.\qed
\end{proof}

\textbf{Acknowledgements:} We want to thank Fran\c{c}ois Gonze and Jarkko Kari for the helpful discussions.

\bibliographystyle{plain}
\bibliography{synchronization}

\end{document}